\documentclass[12pt,english]{paper}
\usepackage[T1]{fontenc}
\usepackage[latin9]{inputenc}
\usepackage{babel}
\usepackage{subfigure}
\usepackage{amsmath}
\usepackage{amsfonts}
\usepackage{amsthm}
\usepackage{mathtools}
\usepackage{amssymb}
\usepackage{dsfont}
\usepackage{float}
\usepackage{setspace}
\usepackage{bbm}
\usepackage{rotating}

\usepackage{color} 

\usepackage{graphicx}
\usepackage[normalem]{ulem}
\usepackage{geometry}
\usepackage{natbib}

\usepackage{hyperref}
\hypersetup{
    colorlinks=false,
    linkcolor=blue,
    filecolor=magenta,      
    urlcolor=cyan,
}
 
\usepackage[colorinlistoftodos, obeyDraft]{todonotes} 

\usepackage{tikz}
\usepackage{pgfplots}

\makeatletter
\theoremstyle{plain}

\theoremstyle{plain}
\newtheorem{lemma}{\protect\lemmaname}
\theoremstyle{plain}
\newtheorem{proposition}{\protect\propositionname}
\theoremstyle{definition}

\theoremstyle{plain}

\theoremstyle{plain}
\newtheorem{corollary}{\protect\corollaryname}
\theoremstyle{plain}

\makeatother

\providecommand{\definitionname}{Definition}
\providecommand{\lemmaname}{Lemma}
\providecommand{\propositionname}{Proposition}
\providecommand{\theoremname}{Theorem}
\providecommand{\assumptionname}{Assumption}
\providecommand{\corollaryname}{Corollary}
\providecommand{\remarkname}{Remark}

\begin{document}

\title{ Auctions with Tokens:  Monetary Policy as a Mechanism Design Choice.\thanks{I'm grateful to John Asker, Rainer B{\"o}hme, Sylvain Chassang, Lin William Cong, Andreas Park, two anonymous referees, the associate editor, conference
participants at the Crypto Asset Lab conference 2021, CSH Workshop on Decentralized Finance 2022, ASSA 2023, seminars at CREST, CY University (Thema) 
for their insightful comments and suggestions.}
}

\author{ Andrea Canidio
\thanks{CoW Protocol, Lisbon, Portugal; andrea@cow.fi.}
}

\maketitle

\vspace*{-.3 cm}

\noindent First version: September 30, 2021. This version: \today. Please check \href{http://andreacanidio.com/research}{here} for the latest version.

\vspace*{-.3 cm}

\begin{abstract}

I study a repeated auction in which payments are made with a blockchain token created and initially owned by the auction designer. Unlike the ``virtual money'' previously examined in mechanism design, such tokens can be saved and traded outside the mechanism. I show that the present-discounted value of expected revenues equals that of a conventional dollar auction, but revenues accrue earlier and are less volatile. The optimal monetary policy burns the tokens used for payment, a practice common in blockchain-based protocols. I also show that the same outcome can be reproduced in a dollar auction if the auctioneer issues a suitable dollar-denominated security. This equivalence breaks down with moral hazard and contracting frictions: with severe contracting frictions the token auction dominates, whereas with mild contracting frictions the dollar auction combined with a dollar-denominated financial instrument is preferred.
\end{abstract}

\maketitle

\section{Introduction}

The well-functioning of a blockchain protocol depends on incentive-compatible behavior among its participants. This is often achieved through protocol-specific blockchain-based tokens. In this respect, blockchain protocols resemble mechanisms studied in the mechanism design literature, particularly those involving privately-issued currencies or ``virtual money.'' For instance, several business schools allocate students to MBA classes by distributing ``points'' that students use to bid for courses (see \cite{hylland1979efficient}, \cite{budish2011combinatorial}, \cite{budish2017course}, \cite{he2018pseudo}). Similarly, the organization Feeding America allocates food to local food banks through auctions conducted using a virtual currency (see \cite{prendergast2017food}, \cite{prendergast2021food}). Blockchain-based tokens, however, differ from these virtual currencies in a crucial respect: they can be saved and traded independently of the mechanism in which they are used. As a result, they constitute a new class of financial instruments and raise questions that do not usually arise in traditional mechanism design. Will some tokens be held rather than used within the mechanism? How does this affect incentives, the equilibrium allocation, and the revenues of the mechanism designer? What constitutes an optimal monetary policy in such environments?

Therefore, analyzing blockchain protocols as mechanisms requires extending the standard mechanism design framework to account for the possibility of issuing blockchain-based tokens. This paper takes a first step in this direction by focusing on a specific mechanism: a repeated auction. Auctions are among the best-understood mechanisms and constitute a natural starting point. They are also widely used in blockchain contexts. For example, most blockchain protocols allocate block space via repeated auctions denominated in the platform's native token (e.g., ETH in Ethereum, SOL in Solana, BTC in Bitcoin). Auctions also play a central role in many blockchain-based applications. A relevant example is the sale of non-fungible tokens (NFTs), particularly those representing ownership of in-game or virtual-world assets. These auctions are typically repeated over time and accept payment exclusively in the platform's native token.\footnote{For instance, virtual land in the metaverse platform ``The Sandbox'' is purchased using its native token SAND. In ``Decentraland,'' land must be purchased using MANA, its native token. For future reference, note that the MANA used in land auctions is then burned (see \url{https://decentraland.org/blog/announcements/the-decentraland-land-auction-has-started}).}


I consider a finite sequence of private-value auctions in which one object is sold per period.\footnote{The companion short paper \cite{canidio2022Auction-pp} considers the case of an infinitely repeated auction. The main issue there is the emergence of financial bubbles on tokens.} In each period, risk-neutral bidders draw valuations independently from a common distribution and submit bids. The auction format then determines the winner and the payments. The good is consumed within the same period. The point of departure from the ``textbook'' model of auctions is that the auctioneer can accept payment in dollars or in a blockchain-based token he creates and initially owns. In the latter case, he also commits to a monetary policy---a rule governing the evolution of the token supply. The auctioneer earns revenue by selling newly created tokens to bidders and reselling the tokens received as payment. Tokens may be held between periods and traded on a financial market, where their price is determined endogenously in equilibrium.

The model, therefore, relies on two main assumptions. First, the auctioneer can credibly commit to a monetary policy. The source of commitment depends on the application. In decentralized platforms (e.g., NFT auctions), commitment is enforced through smart contracts, since both the application and the token are implemented as smart contracts. In blockchain protocols (e.g., Ethereum's block-building auction), the monetary policy is embedded in the protocol's open-source software. Although protocol updates are possible, they are rare and require broad validator consensus. Commitment thus arises from decentralized governance, which prevents unilateral deviations. This justifies treating monetary policy as a credible design choice implementable through code. Second, the model assumes that tokens can be easily issued and traded. Blockchain-based tokens (e.g., ERC-20 tokens on Ethereum) can be created with minimal technical knowledge and at negligible cost. Once created, they can be listed on decentralized exchanges such as Uniswap, where they can be exchanged with other assets, including stablecoins (i.e., tokens pegged to the value of the dollar).


The model applies directly when the auction designer is also the seller, as in the NFT auctions mentioned above. It also accommodates settings in which the designer (and token issuer) differs from the seller, provided the seller's reservation value is common knowledge (so that Myerson-Satterthwaite's impossibility result does not apply). In that case, the designer functions as a platform.  A leading example is the allocation of block space in blockchain protocols: Ethereum's core developers set the auction format, payment token, and monetary policy, whereas miners or validators---who have no intrinsic value for block space---are the sellers. Similarly, on the demand side, it is possible to distinguish between short-lived bidders who acquire tokens only for the current auction and cash-abundant speculators or investors who hold tokens across periods.  For clarity, I adopt the simplest interpretation: the auctioneer is simultaneously designer and seller, and bidders may carry token balances intertemporally.  I return to the richer role-separation interpretation when it matters for specific results.

I show that in the token-based auction, when realized valuations in a given period are low relative to expected future valuations, bidders may purchase tokens for speculative purposes rather than for bidding. The reason is that when valuations are low, the demand for tokens \textit{for bidding} is low, creating an arbitrage opportunity: bidders may purchase tokens, not use them for bidding in that period, and sell them (or use them) in future periods.  This speculative demand increases the token price in the current period and raises the auctioneer's expected revenue above that in the dollar-based auction. However, future revenues fall below that of the dollar-based auction since today's speculators compete with the auctioneer in tomorrow's token market. Because speculative demand depends on expected future valuations, it effectively converts uncertain future revenues into certain current ones.

Building on this result, I then establish a form of revenue equivalence: the present-discounted value of expected revenues is identical across auction formats, whether payments are made in dollars or in tokens. However, in any given period, revenues may differ. In particular, revenues accrue earlier and are less volatile when the auction uses tokens than when it uses dollars. The precise timing and volatility of revenues in the auction with tokens depend on the monetary policy implemented by the auctioneer. A particularly relevant policy is one in which the auctioneer sells all tokens in the first period and subsequently destroys (i.e., burns) the tokens received as payment. Under this policy, the auctioneer earns all revenues in period 1, including the present-discounted value of expected revenues from period 2 onward.\footnote{In the model, the initial token sale occurs after period-1 valuations are drawn. As a result, period 1 is the only source of revenue volatility. A straightforward extension of the model allows the auctioneer to sell tokens before the realization of period-1 valuations, in which case all revenue uncertainty can be eliminated.}

It follows that by appropriately designing a token-based auction, the auctioneer can fully front-load revenues and eliminate (almost) all risk. I then show that, in the absence of financial or contractual frictions, the same outcome can be replicated by conducting an auction with dollars in which the auctioneer transfers future cash flows to investors by issuing equity. This yields an equivalence result: when frictions are absent, the optimal token-based auction is equivalent to a dollar-based auction accompanied by equity issuance.


The second part of the paper shows that in the presence of frictions, the equivalence between token-based and dollar-based auctions may break down. I illustrate this by considering a simplified version of the model and introducing two frictions: (i) costly, non-contractible effort and (ii) revenue misappropriation, that is, the possibility of diverting revenues by ``running away with the till,'' albeit at a cost. Intuitively, when issuing tokens, the auctioneer implicitly commits to deliver the object in the future. This commitment may be imperfect if the auctioneer can shirk, exert low effort, and supply an object of lower value to bidders/token-holders. Alternatively, the auctioneer may conduct a dollar-based auction while pledging future cash flows to investors via a standard financial instrument. The key observation is that such a pledge is credible only if the auctioneer can commit both to deliver an object of high value \textit{and} to refrain from diverting revenues generated by its sale.


Hence, whether to issue tokens or a traditional financial instrument hinges on a trade-off. On the one hand, revenue misappropriation is a concern only when issuing a traditional financial instrument. On the other hand, depending on the contracting environment, traditional financial instruments can specify state-contingent payments that tokens cannot replicate. Such payments can motivate non-contractible effort whenever a contractible proxy (realized revenues) is available. I show that tokens are preferred when the cost of hiding revenues is sufficiently low, because misappropriation is not an issue. Conversely, a dollar-based auction with a traditional financial instrument is preferred when diverting revenues is sufficiently costly. The cost of misappropriation proxies the degree of contractual incompleteness, as it determines the set of implementable, incentive-compatible contracts. The resulting implication is that tokens are preferred when the contracting space is limited, whereas traditional financial instruments are preferred when the contracting space is sufficiently rich.

This paper makes several contributions. One is to provide a theoretical rationale for burning tokens, a practice commonly observed in blockchain-based auctions. Token burning has been widely debated in the Ethereum community, particularly in connection with EIP-1559, a recent protocol upgrade that introduced a ``base fee'': a minimum payment required for a transaction to be included in a block, adjusted dynamically based on demand for block space.\footnote{See the original proposal at \url{https://eips.ethereum.org/EIPS/eip-1559}.} The goal of the base fee was to make transaction costs more predictable for users. Simultaneously, EIP-1559 mandated burning the base fee to deter manipulation by miners or validators (see \citealp{roughgarden2020transaction} for a formal argument). Several commentators have noted, however, that if deterring such manipulation were the sole objective, the base fee could instead be redirected to a dedicated fund. This paper contributes to the debate by showing that token burning also has desirable monetary policy implications.

This paper contributes to the ongoing discussion on whether tokens should be classified as securities. In the model, tokens grant access to a service and are therefore considered ``utility tokens.'' Their pricing and valuation logic differ from that of equity (see, in particular, \citealp{prat2019fundamental}), which led some to argue that utility tokens are not securities. However, in the benchmark model without frictions, the equilibrium value of tokens coincides with that of equity. The analysis also provides a rationale for preferring tokens over equity. When the contracting environment is sufficiently constrained, tokens are optimal because revenue misappropriation is a concern with equity but not with tokens. A key prediction is therefore that token issuance should be more prevalent in jurisdictions with weaker institutions, for example, where enforcement through the court system is less effective.

\subsection*{Literature review}


This paper contributes to a nascent literature exploring the intersection of blockchain and mechanism design (see \citealp{townsend2020distributed} for an overview). The most studied aspect of this intersection is the use of blockchain-based smart contracts to generate commitment and implement various mechanisms (see \citealp{holden_malani_2021}, \citealp{gans2019fine}, \citealp{lee2021optimal}, \citealp{brzustowski2021smart}, \citealp{townsend-zhang}). A related and growing literature applies insights from mechanism design to improve blockchain protocols (see, for example, \cite{gans2022mechanism}, \cite{gans2023solomonic}, \cite{Gans2023}, \cite{capponi2021adoption}, and \cite{canidio-danos}). Unlike those papers, my focus is not the implementation of the auction; that is, the auctions I study could be centralized or implemented through a smart contract. Rather, this paper highlights that the ability to issue blockchain-based tokens with money-like properties expands the auctioneer's design space. In this respect, it is related to \cite{kocherlakota1998money}, who argues that ``money is memory.'' Blockchain, in a sense, replaces the auctioneer's internal ledger. This yields two benefits: first, increased efficiency in recording transactions, particularly in secondary markets; and second, credible commitment to a specific monetary policy.

A growing theoretical literature studies firms' incentives to issue blockchain-based tokens. These tokens may represent a pre-sale of future output, a medium of exchange accepted exclusively by the issuer, or a claim on future revenues or profits. Several papers show that, in the presence of network externalities, token issuance can help overcome coordination failures (see \citealp{sockin2018model}, \citealp{cong2018tokenomics}, \citealp{bakos2018role}, \citealp{li2018initial}). Other contributions view token sales as a novel means of raising capital and financing the development of a product or platform (see \citealp{RePEc:nbr:nberwo:24418}, \citealp{malinova2018tokenomics}, \citealp{canidio2018financial}, \citealp{bakos2019funding}, \citealp{goldstein2019initial}, \citealp{cong2020token}, \citealp{canidio2020cryptotokens}, \citealp{gryglewicz2021optimal}, \citealp{GARRATT2021104171}, \citealp{doi:10.1287/mnsc.2020.3754}). In contrast, in the model I study, the auctioneer faces no financing needs and network externalities are absent. Nonetheless, token issuance allows the auctioneer to alter the timing and risk profile of revenues---factors that play a central role in shaping investment incentives and the creation of new ventures.

Two classical results from auction theory underpin most of the analysis. The first is the revenue equivalence theorem, which states that if bidders are risk-neutral and valuations are independently and identically distributed (i.i.d.), then any standard auction format yields the same expected revenue and induces the same expected payment from each bidder. An auction is considered ``standard'' if the object is always allocated to the highest bidder and a bidder with the lowest possible valuation always earns zero.\footnote{\cite{vickrey1961counterspeculation} developed an early special case of the revenue equivalence theorem. The formulation presented here follows \cite{klemperer1999auction}, which summarizes results from \cite{myerson1981optimal} and \cite{riley1981optimal}. For a generalization, see \cite{milgrom2002envelope}.} 
The model assumes that the conditions of the revenue equivalence theorem are satisfied. The second result concerns optimal auction design (see \citealp{myerson1981optimal}, \citealp{bulow1989simple}, \citealp{bulow1996auctions}, and \citealp{klemperer1999auction}). Specifically, under the assumptions of the revenue equivalence theorem, if the distribution of valuations is well-behaved (i.e., continuous and atmoless) there exists a reservation price that maximizes the auctioneer's revenues, which depends on the distribution of valuations and not on the auction format.

\section{The model}\label{sec: model}

\todo[inline]{introduce a discount factor specific to the auctioneer}

Consider a private-value auction for a single object, repeated $T \geq 1$ times. There are $n \geq 2$ ex-ante identical bidders and one auctioneer. Bidders are risk-neutral and cash-abundant, in the sense that their liquidity constraints are never binding. The auctioneer has a weakly concave per-period utility function $U(\cdot)$. Both the auctioneer and the bidders share a common discount factor $\beta \in (0,1)$ and may invest in a risk-free asset yielding a gross return $R \geq 1$ per period. For simplicity, assume $R = \frac{1}{\beta}$, so that $R$ corresponds to the steady-state return in the Ramsey-Cass-Koopmans growth model without population or productivity growth. The auctioneer's initial assets are $w_1 \geq 0$. In this section, I assume that the auctioneer can save only via the risk-free asset, so that $w_t \geq 0$ for all $t \leq T$, where $w_t$ denotes asset holdings at the beginning of period $t$. Section~\ref{sec:borrowing and saving} extends the analysis to allow the auctioneer to issue equity. 

In period $0$, the auctioneer chooses whether to accept payments in fiat currency (for simplicity, dollars) or in tokens. If tokens are used, the auctioneer creates an initial stock $M_1$ and announces a monetary policy, that is, a rule governing the evolution of the token supply (defined more formally below). The auction format and monetary policy are fixed at this stage. The auctioneer sells a single object under the specified auction format in each subsequent period $t \in \{1, \dots, T\}$. At the beginning of each period $t \geq 1$, each bidder draws a valuation $v_{i,t} > 0$ from a continuous and atomless distribution with cumulative distribution function $F(v)$, probability density function $f(v)$, and support $[\underline{v}, \overline{v}]$. Each $v_{i,t}$ is privately known to bidder $i$, while the distribution $F(v)$ is common knowledge. Valuations are private and independently drawn across bidders. The object sold in each period has zero value to the auctioneer.

Three clarifications on the model setup are in order. First, bidders play two roles: they participate in each auction and may hold tokens or other assets between periods. An equivalent formulation would feature short-lived bidders who participate only in the auction, alongside long-lived speculators who hold assets across time. Second, the model assumes that the auction designer and the seller are the same agent. These roles could be separated, provided the seller's reservation price is known. In that case, the designer would resemble a platform: it determines the mechanism, selects the payment currency, and imposes fees or subsidies. In equilibrium, the seller trades at her reservation price, and all additional revenues accrue to the designer. 
Third, the assumption $R = \frac{1}{\beta}$ is adopted for tractability. It implies that bidders are indifferent between immediate consumption and investing in the risk-free asset, allowing this aspect of the problem to be abstracted away. Relaxing this assumption would introduce additional considerations regarding the existence of the equilibrium in the risk-free asset market. It would also complicate the notation because bidders hold tokens between periods only if their expected return exceeds $\max\left\{ \frac{1}{\beta}, R \right\}$.

If the auctioneer accepts dollars, the sequence of events in each period follows standard timing. After drawing valuations, each bidder submits a message $m_{i,t} \in \mathbb{R}_+$, interpreted as a bid. Based on the message profile and the auction format specified in advance, the auctioneer determines the winner and assigns payments $b_{i,t} \leq m_{i,t}$ for each bidder. Each bidder then pays the auctioneer $b_{i,t}$ dollars, and the winning bidder receives and consumes the object. The period-$t$ payoff of the winning bidder is $v_{i,t} - b_{i,t}$, while the payoff of each losing bidder is $-b_{i,t}$. The auctioneer's revenues in period $t$ are $\sum_{i=1}^n b_{i,t}$. At the end of the period, the auctioneer invests $w_{t+1}/R$ in the risk-free asset (which becomes $w_{t+1}$ in period $t+1$) and consumes $w_t + \sum_{i=1}^n b_{i,t} - w_{t+1}/R$.

If the auctioneer instead uses tokens, the timeline of each period $t \in \{1, \dots, T\}$ is as follows:
\begin{itemize}
    \item At the start of the period, each bidder draws a valuation $v_{i,t}$ from $F(v)$ and submits a message $m_{i,t} \in \mathbb{R}_+$, interpreted as a bid \textit{in dollars}. At this point, both the auctioneer and bidders may hold tokens carried over from previous periods. Let $A_t \geq 0$ denote the auctioneer's token holdings at the beginning of period $t$, and $a_{i,t} \geq 0$ denote those of bidder $i$. By assumption, $A_1 = M_1$ and $a_{i,1} = 0$ for all $i$.

    \item Based on the submitted bids and the pre-specified auction format, the auctioneer determines the winner and assigns payments $b_{i,t} \leq m_{i,t}$ to each bidder. These payments are denominated in dollars \textit{but must be settled in tokens}.\footnote{Bids and payments are expressed in dollars for notational consistency with the dollar-based auction. The results are identical whenever bids are expressed in tokens.}

    \item A frictionless, anonymous market for tokens opens, in which both the auctioneer and bidders participate as price takers. Let $p_t$ denote the equilibrium token price; let $q_{i,t}$ be the equilibrium demand for tokens by bidder $i$ and $Q_t$ that of the auctioneer. Demands may be positive or negative, indicating net purchases or net sales. Market clearing requires
    \[
    A_t + \sum_{i=1}^n a_{i,t} = Q_t + \sum_{i=1}^n q_{i,t}.
    \]
    Moreover, since each bidder must acquire enough tokens to make the payment, it must be that
    \[
    q_{i,t} \geq \frac{b_{i,t}}{p_t} - a_{i,t}.
    \]

    \item Each bidder transfers $\frac{b_{i,t}}{p_t}$ tokens to the auctioneer. The winner receives the object and consumes it. At this point, bidder $i$ holds $a_{i,t} + q_{i,t} - \frac{b_{i,t}}{p_t}$ tokens, and the auctioneer holds $A_t + Q_t + \sum_{i=1}^n \frac{b_{i,t}}{p_t}$ tokens.

    \item For the winner of the auction, its period-$t$ payoff is $v_{i,t} - p_t q_{i,t}$. The other bidder's period-$t$ payoff is instead $-p_t q_{i,t}$. The auctioneer earns period-$t$ revenues of $-p_t Q_t$. She invests $w_{t+1}/R$ in the risk-free asset and consumes $w_t - p_t Q_t + w_{t+1}/R$.

    \item The token supply evolves according to the monetary policy announced at the beginning of the game. Two time-varying policy instruments are considered: a uniform growth (or shrinkage) factor $\tau_t \geq -1$ applied to all tokens, and a factor $\sigma_t \geq -1$ applied specifically to tokens used for bidding. At the start of period $t+1$, bidder $i$ holds
    \[
    a_{i,t+1} = (1+\tau_t) \left(a_{i,t} + q_{i,t} - \frac{b_{i,t}}{p_t} \right),
    \]
    and the auctioneer holds
    \[
    A_{t+1} = (1+\tau_t) \left( A_t + Q_t + (1+\sigma_t) \sum_{i=1}^n \frac{b_{i,t}}{p_t} \right).
    \]
    Hence, the total token supply at the beginning of period $t+1$ is given by $M_{t+1} = A_{t+1} + \sum_{i=1}^n a_{i,t+1}$.
\end{itemize}

The monetary policy considered here may be unconventional, but it is straightforward to implement with blockchain-based tokens. The distinction between the growth rate of tokens used for bidding and other tokens is inspired by staking mechanisms. In staking, agents who ``lock'' tokens---typically by refraining from using or transferring them---receive additional tokens as a reward. In the model context, a positive staking reward arises when $\sigma_t < 0$: tokens not used for bidding grow at a higher rate than those used, creating an incentive to hold rather than spend.

The auction described above is deliberately designed to remain as close as possible to a traditional auction with dollars: bids and payments are denominated in dollars but must be settled using tokens. Alternative formulations are possible. For instance, bidders could be required to submit bids directly in tokens, which may then be partially refunded after the winner is determined.\footnote{If bidders must secure tokens \emph{before} submitting bids, the token price can embed information about bidders' realized valuations. The resulting token-market equilibrium is then a rational-expectations equilibrium, whose existence is not guaranteed---at least absent additional elements such as exogenous noise or a randomization device. If we sidestep these issues and assume that the market price does not reveal information relative to bidders' valuations, the analysis in the main text survives this extension: each auction becomes a ``partial'' all-pay auction (partial because tokens retain some resale value), and the standard revenue-equivalence logic still applies.} Moreover, while the analysis considers only two monetary policy instruments, more general policies are feasible. 



\section{Equilibrium}
I first consider the auction with dollars and then the auction with tokens. 

\subsection{Auction with dollars}

When the auction uses dollars, standard results from auction theory apply.\footnote{See, for example, Section 4 and Appendix B of \cite{klemperer1999auction}.} In particular, in every period, the revenue equivalence theorem holds: all standard auction formats yield the same expected revenues. Moreover, expected revenues are maximized by an appropriate reservation price is $\underline b$, which depends exclusively on the distribution of valuations. Focusing on the second-price auction for ease of derivations, the maximum expected revenue in each period is
\[
k \equiv \mathbb{E}\left[\max\left\lbrace v_{\text{Max-1},t}, \underline b\right\rbrace | v_{\text{Max},t} \geq \underline b\right]\cdot \mbox{pr}\left\lbrace v_{\text{Max},t} \geq \underline b\right\rbrace,
\]
where $v_{\text{Max},t} \equiv \max_i \{ v_{i,t} \}$ is the highest realized valuation in period $t$, and $v_{{Max-1},t}\equiv \max_{i\neq Max, t}\{v_{i,t}\}$ is the second-highest realized valuation. Each bidder's expected per-period payoff is
\[
g \equiv \mathbb{E}[\max\{v_i - \max\left\lbrace v_{\text{Max-1},t}, \underline b\right\rbrace, 0\}].
\]
Hence, from the perspective of period 1, the present-discounted value of expected revenues in the auction with dollars is
\[
\Pi_{\text{USD}} = k \sum_{t=1}^T \beta^{t-1} = \frac{1 - \beta^T}{1 - \beta} \, k,
\]
and the present-discounted value of participating in the auction as a bidder is
\[
u_{\text{USD}} = g \sum_{t=1}^T \beta^{t-1} = \frac{1 - \beta^T}{1 - \beta} \, g.
\]

For a given auction format, the auctioneer's utility is
\[
U_{\text{USD}} \equiv \max_{w_2 \geq 0, \dots, w_T \geq 0,\, w_{T+1} = 0} \left\{ \sum_{t=1}^T \beta^{t-1} \, \mathbb{E} \left[ U \left( \sum_{i=1}^n b_{i,t} + w_t - \frac{w_{t+1}}{R} \right) \right] \right\},
\]
where $U(\cdot)$ denotes the auctioneer's per-period utility function and $\mathbb{E} [U(\cdot)]$ represents its expectation. If the auctioneer is risk-averse (i.e., $U(\cdot)$ is strictly concave), not all standard auction formats maximize expected utility, since the variance of revenues also affects the expected utility. However, to compare the dollar-based auction with the token-based alternative in the next section, it suffices to derive an upper bound on the auctioneer's utility. This is provided in the following lemma (its proof is omitted).

\begin{lemma}\label{lem:auction USD}
Consider a given auction format, an initial asset level $w_1 \geq 0$, and realized period-1 revenues $\sum_{i=1}^n b_{i,1}$. Define $\{w_2^*, \dots, w_T^*\}$ as the unconstrained optimal sequence of asset holdings in the absence of risk, that is,
\[
\{w_2^*, \dots, w_T^*\} \equiv \arg\max_{w_2, \dots, w_{T+1}=0} \left\{ U\left( \sum_{i=1}^n b_{i,1} + w_1 - \frac{w_2}{R} \right) + \sum_{t=2}^T \beta^{t-1} U\left( k + w_t - \frac{w_{t+1}}{R} \right) \right\}.
\]
Then,
\[
U_{\text{USD}} \leq \mathbb{E} \left[ U\left( \sum_{i=1}^n b_{i,1} + w_1 - \frac{w_2^*}{R} \right) + \sum_{t=2}^T \beta^{t-1} U\left( k + w_t^* - \frac{w_{t+1}^*}{R} \right) \right],
\]
where the expectation is taken over the realization of period-1 revenues $\sum_{i=1}^n b_{i,1}$. The inequality is strict if $U(\cdot)$ is strictly concave.
\end{lemma}

Lemma~\ref{lem:auction USD} establishes that, for any auction format, the auctioneer's utility is bounded above by the utility he would attain if all risk beyond period 1 and the borrowing constraints were removed. The inequality is strict when $U(\cdot)$ is strictly concave. Intuitively, a risk-averse auctioneer prefers receiving the expected per-period revenue $k$ with certainty rather than being exposed to revenue volatility. Moreover, the concavity of the utility function implies a preference for smoothing consumption across periods. Because achieving the optimal consumption profile may require borrowing, removing credit constraints enhances utility, regardless of whether uncertainty is present.

\subsection{Auction with Tokens}

This section begins with the characterization of the equilibrium price of tokens in period $t$, as a function of the expected token price in the following period, $p^e_{t+1}$, and the realization of bids.

\begin{lemma}\label{lem: initial}
Let $p^e_{t+1} \equiv \mathbb{E}[p_{t+1}]$ denote the expected token price in period $t+1$, conditional on information available in period $t$. Consider a given value of $p^e_{t+1}$ and a realization of period-$t$ valuations (and hence a profile of bids $\{m_{i,t}\}_{i=1}^n$ and payments $\{b_{i,t}\}_{i=1}^n$). In equilibrium, the total demand for tokens in period $t$ is
\[
\frac{\sum_{i=1}^n b_{i,t}}{p_t} + S_t,
\]
where $S_t \geq 0$ denotes the speculative demand for tokens (i.e., tokens not used for bidding in period $t$) and is given by
\[
S_t \equiv \max \left\{ M_t - \frac{\sum_{i=1}^n b_{i,t}}{\beta (1 + \tau_t) p^e_{t+1}},\, 0 \right\}.
\]
The equilibrium price of tokens in period $t$ is
\begin{equation}\label{eq: equilibrium prices}
p_t = \max \left\{ \frac{\sum_{i=1}^n b_{i,t}}{M_t},\, \beta (1 + \tau_t) p^e_{t+1} \right\}.
\end{equation}
\end{lemma}

A key implication of Lemma~\ref{lem: initial} is that some tokens may be purchased not for bidding, but for speculative purposes. This arises in equilibrium when realized valuations are sufficiently low. In such cases, if token demand were driven solely by the need to settle bids, the equilibrium price would be $p_t < \beta (1 + \tau_t) p^e_{t+1}$. However, this cannot be an equilibrium: under this condition, the return from holding tokens  $\beta^{-1} (1 + \tau_t) p^e_{t+1} / p_t$ would exceed that of the risk-free asset, which yields a present-discounted return of $\beta R = 1$. To eliminate arbitrage, speculative demand emerges, pushing the token price to at least $\beta (1 + \tau_t) p^e_{t+1}$. Thus, the equilibrium price in period $t$ has a lower bound determined by the value of holding tokens as a speculative asset.

Lemma~\ref{lem: initial} can also be used to characterize bidders' incentives given a sequence of token prices. If $p_t > \beta (1 + \tau_t) p^e_{t+1}$, bidders strictly prefer not to hold tokens between periods and only invest in the risk-free asset, implying $a_{i,t+1} = 0$. If instead $p_t = \beta (1 + \tau_t) p^e_{t+1}$, bidders are indifferent between holding tokens or the risk-free asset between periods $t$ and $t+1$, including holding zero tokens. In both cases, given token prices $p_t$ and $p^e_{t+1}$ and the profile of bids, bidder $i$'s period $t$ payoff is
\begin{equation}\label{eq: utility}
\begin{cases}
v_{i,t} + p_t a_{i,t} - b_{i,t} + \beta u_{t+1}(0), & \text{if } i \text{ wins}, \\
p_t a_{i,t} - b_{i,t} + \beta u_{t+1}(0), & \text{otherwise},
\end{cases}
\end{equation}
where $u_{t+1}(a_{i,t+1})$ denotes the expected continuation utility from period $t+1$ onward, as a function of the token holdings at the start of that period. It follows that, for any given auction format, the bidders' incentives to bid are the same with or without tokens. As a result, the revenue equivalence theorem continues to hold: any standard auction format yields the same expected payments to the auctioneer, i.e., $\mathbb{E}[\sum_{i=1}^n b_{i,t}] = k$. Accordingly, each bidder's expected per-period payoff is again $g$.

Importantly, equation~\eqref{eq: utility} implies that a bidder's payoff depends on both the expected payoff from participating in the auction, $g$, and the expected value of the tokens held at the beginning of the period, $p^e_t \cdot a_{i,t}$.\footnote{This can be interpreted as each bidder selling their token holdings at the market price $p_t a_{i,t}$ while simultaneously purchasing tokens to make the payment $b_{i,t}$. The expected cost of bidding is part of the expected payoff from the auction.} Taking the expectation of~\eqref{eq: utility}, and using the fact that the expected payoff from the auction is $g$, bidder $i$'s expected continuation utility can be written as
\[
u_t(a_{i,t}) = p^e_t a_{i,t} + g \sum_{s=t}^T \beta^{s-t} = p^e_t a_{i,t} + g \cdot \frac{1 - \beta^{T - t + 1}}{1 - \beta}.
\]
In particular, since $a_{i,1} = 0$, the expected continuation utility from the period-1 perspective is the same as in the auction with dollars. However, for $t \geq 2$, token holdings $a_{i,t}$ may be positive, so the intertemporal profile of utility may differ between the token-based and dollar-based auctions.

Having characterized bidders' behavior, it is now possible to derive the auctioneer's expected revenues.
\begin{proposition}[Expected Revenues]\label{prop:revenues}
Consider a given sequence of equilibrium expected prices, i.e., a sequence satisfying equation~\eqref{eq: equilibrium prices}. At the beginning of each period $t \leq T$, the present-discounted value of the auctioneer's expected future revenues is
\[
\Pi_{\text{Tokens},t}(A_t) = \frac{1 - \beta^T}{1 - \beta} \, k - p^e_t (M_t - A_t).
\]
\end{proposition}

The key insight is that speculative demand in period $t$ raises the price of tokens, thereby increasing the auctioneer's revenues in that period. However, this comes at the cost of reduced revenues in future periods, as speculators compete with the auctioneer in the token market in period $t+1$. Proposition~\ref{prop:revenues} shows that, in expectation, these effects offset each other: the total present-discounted value of revenues remains unchanged relative to the auction with dollars. In particular, the auctioneer's expected revenues at the start of the game are $\Pi_{\text{Tokens},1}(M_1) = \frac{1 - \beta^T}{1 - \beta} \, k$, as in the dollar-based auction. In later periods, however, expected continuation revenues may fall below the dollar benchmark by an amount equal to the value of tokens not held by the auctioneer. In addition to influencing the time-profile of revenues, speculative demand also influences their volatility. The reason is that the speculative demand for tokens depends on the future expected price and, therefore, transforms uncertain future revenues into certain present revenues.

Crucially, the incentive to purchase tokens for speculative purposes---and thus the time profile and variability of revenues---depends on the monetary policy specified by the auctioneer. When $\sigma_t$ is sufficiently large for all $t \leq T$, the auctioneer creates and retains a large number of tokens in each period. The resulting inflation eliminates speculative demand, and the auctioneer earns $\sum_{i=1}^n b_{i,t}$ in every period, as in the dollar-based auction. As $\sigma_t$ decreases, speculative demand increases, therefore front-loading revenues and reducing their volatility. The case $\sigma_t = -1$ is particularly relevant: all tokens are sold to bidders in period 1, and then progressively used to pay the auctioneer, who then destroys them. The following proposition characterizes the revenues under this policy. It follows directly from Proposition~\ref{prop:revenues} and by noting that revenues accrue after the period-1 valuations are drawn.
\begin{proposition}[Revenues when $\sigma_t = -1$ for all $t \leq T$]\label{prop: revenue finite}
If $\sigma_t = -1$ for all $t \leq T$, then the auctioneer earns revenues exclusively in period 1, and those revenues are
\begin{equation}
p_1 M_1 = \sum_{i=1}^n b_{i,t} + \beta \cdot \frac{1 - \beta^{T-1}}{1 - \beta} \, k.
\end{equation}
\end{proposition}

When the auctioneer destroys all tokens received as payment, all revenues are realized in period 1. While period-1 revenues remain stochastic, the present-discounted value of revenues from periods $t \geq 2$ is earned with probability one.\footnote{A natural extension of the model allows the auctioneer to sell tokens already in period 0. In that case, period-1 revenue risk can be eliminated, as the entire present-discounted revenue stream is realized in period 0 with certainty.} Since all risk after period 1 is eliminated and all revenues are front-loaded, the auctioneer can achieve optimal consumption smoothing through saving. For any realization of $\sum_{i=1}^n b_{i,1}$, the auctioneer attains utility
\[
U\left( \sum_{i=1}^n b_{i,1} + R w_1 - w^*_2 \right) + \sum_{t=2}^T \beta^{t-1} U\left( k + R w^*_t - w^*_{t+1} \right),
\]
where $\{w^*_2, \dots, w^*_T\}$ is defined in Lemma~\ref{lem:auction USD}. Lemma~\ref{lem:auction USD} then implies the following corollary.

\begin{corollary}\label{prop:preferred auction}
The token-based auction with $\sigma_t = -1$ for all $t \leq T$ yields strictly higher utility to the auctioneer than any other token-based auction if $U(\cdot)$ is strictly concave. It is also  preferred to the auction with dollars, strictly so if $U()$ is strictly concave.
\end{corollary}

\todo[inline]{I need to complete the argument by showing that all possible auction formats expose the auctioneer to the variability of period 1 revenues}

\subsection{Discussion: Traditional Financial Instruments}\label{sec:borrowing and saving}

A natural question is whether the optimal auction with tokens can be replicated by holding an auction with dollars while simultaneously issuing a traditional financial instrument.

The answer is affirmative. Suppose the auctioneer conducts a dollar-based auction and issues equity in period 1. Investors pay the auctioneer for the right to receive all future cash flows. For ease of comparison with the token-based case, assume that equity is sold at the end of period 1, just before the auctioneer consumes.

In equilibrium, investors must be indifferent between purchasing equity or not. Since equity yields the revenues generated from period 2 onward, the auctioneer receives
\[
\beta \cdot \frac{1 - \beta^{T - 1}}{1 - \beta} \, k
\]
from the equity sale. Hence, total revenues in period 1---including proceeds from the auction and equity---are
\[
\sum_{i=1}^n b_{i,1} + \beta \cdot \frac{1 - \beta^{T - 1}}{1 - \beta} \, k.
\]
The auctioneer earns no further revenues after period 1. Thus, the resulting equilibrium is identical to the token-based auction with $\sigma_t = -1$ for all $t \leq T$.

\section{Extension: pledging an object vs.\ pledging cash}

When issuing tokens, the auctioneer commits to delivering an object to investors at a future date. By contrast, issuing a traditional financial instrument such as equity commits the auctioneer to transferring the \textit{cash} proceeds obtained from selling that object. In the frictionless benchmark analyzed above, there is perfect commitment, and pledging to deliver the object is equivalent to pledging to deliver cash.

More in general, however, the two types of pledges may not be equivalent. On the one hand, any friction reducing the ability to pledge the object, such as non-contractible effort, necessarily also undermines the commitment to transfer the cash generated from the sale of the object. At the same time, certain frictions impede pledging cash but not the object itself---for example, when cash is easier to conceal. This asymmetry suggests that tokens should be preferred over traditional financial instruments. On the other hand, traditional financial instruments enable state-contingent cash transfers that tokens may not be able to replicate; these state-contingent transfers can incentivize effort even when effort itself is not directly contractible.

To explore the interaction between these frictions, I modify the previous model by introducing two new features: non-contractible effort and the possibility of revenue misappropriation. The auction now runs for $T=2$ periods, and the risk-free asset yields a gross return of $R=1$. The auctioneer remains risk-neutral and cash-abundant, but, unlike before, is now less patient than investors: investors do not discount future payoffs, whereas the auctioneer discounts future payoffs using a factor $\beta<1$.\footnote{The results derived earlier are robust to this change: if fully front-loading revenues by setting $\sigma_t=-1$ for all $t$ is optimal when the auctioneer is equally patient as investors, it remains optimal when the auctioneer is less patient than investors.}

The timing of the first period remains unchanged from the earlier model. However, at the beginning of the second period, the auctioneer now exerts observable but non-contractible effort $e \geq 0$ to enhance the quality of the object. Effort uniformly shifts the distribution of bidders' valuations, such that if bidder $i$ draws valuation $v_{i,2}$ from the distribution $F(x)$, his utility from consuming the object becomes $v_{i,2} + e$.\footnote{Because $e$ is observable, the auction remains a private-value setting despite the valuations having a common additive component.} Exerting effort carries a quadratic cost of $\frac{e^2}{2}$ for the auctioneer. Additionally, after second-period revenues are realized but before payments are made to investors, the auctioneer may misappropriate these revenues. Specifically, by incurring a cost $c \cdot f(x)$, the auctioneer can consume revenues $x$ that otherwise would have to be transferred to investors, where $f(0)=0$, $f(x)>0$ for $x>0$ and $f'(x)\geq 0$. The parameter $c \geq 0$ measures the ease of revenue misappropriation. A natural interpretation is that the auctioneer can ``run away with the till'' containing $x$ by paying the cost $c\cdot f(x)$.

\paragraph{First best}

In the first-best, the effort level equates its marginal cost with its marginal benefit, resulting in $e^{**}=1$. Because the auctioneer discounts future payoffs more heavily than investors, in period~1 investors pay the auctioneer an amount equal to the expected second-period revenues and subsequently recoup this investment in period~2. Consequently, the auctioneer fully front-loads all revenues into the first period.

\paragraph{Auction with dollars and outside investors.}

Suppose the auctioneer accepts payments in dollars and issues a contingent security to outside investors. The contingent security specifies a payment from investors to the auctioneer in period~1, denoted by $y_1$, and a contingent period-2 payment from the auctioneer to investors, denoted by $y_2(\sum_i b_{2,i}+e)$. These payments must satisfy the investors' rationality constraint:
\[
E\left[y_2\left(\sum_i b_{2,i}+e\right)\right]=y_1.
\]

Due to the possibility of ``running away with the till,'' in period~2 the auctioneer must always receive a large-eough payoff. Thus, the incentive compatibility constraint is:
\begin{align*}
\sum_i & b_{2,i}+e - y_2\left(\sum_i b_{2,i}+e\right) \geq \\ & \max_{x\leq \sum_i b_{2,i}+e} \left\lbrace \sum_i b_{2,i}+e -x - y(\sum b_{i,2}+e-x) + x - c\cdot f(x)\right\rbrace, \quad \forall \sum_i b_{2,i}, e,
\end{align*}
where $x$ is the optimal amount of revenue misappropriation.

An important observation is that if misappropriating revenues is costless ($c=0$), the incentive compatibility constraint immediately implies $y_2(\sum_i b_{2,i}+e)=0$: no outside investment is feasible. In this scenario, the auctioneer's choice of effort solves:
\[
\max_{e \geq 0} \left\{ k+e - \frac{e^2}{2} \right\},
\]
yielding the efficient level $e^*=1$. However, the timing of revenues remains inefficient because the auctioneer receives and consumes revenues in both periods (remember that the auctioneer is more impatient than investors, and therefore, in the first-best scenario, she consumes only in the first period). 

In contrast, if $c$ is sufficiently large, the first-best outcome can be attained. To see why, note that if the auctioneer exerts the efficient effort level $e^{**}=1$, realized revenues always exceed $\underline v +1$. Thus, the investment contract can require the auctioneer to transfer all period-2 revenues to investors and impose a penalty whenever period-2 revenues fall below $\underline v +1$. Provided this penalty is sufficiently severe and feasible (which occurs whenever $c$ is sufficiently large), efficient effort is induced. Consequently, investors capture all period-2 revenues, satisfying the investors' rationality constraint $E[y_2(\sum_i b_{2,i}+e)] = y_1 = k+1$ and restoring the first-best allocation.

\paragraph{Auction with tokens}

First, note that the possibility of revenue misappropriation (and thus $c \cdot f(x)$) plays no role in the auction with tokens. Without loss of generality, I normalize the initial stock of tokens to $M_1=1$.

In period~2, speculative demand is zero for any level of effort. Hence, for given effort level $e$, the period-2 price of tokens is $\frac{k+e}{M_2}$ where $M_2$ is the total stock of tokens at the beginning of period 2. Thus, the auctioneer chooses the effort level to maximize:
\[
\max_{e\geq 0} \left\{ \alpha (k+e)-\frac{e^2}{2} \right\},
\]
where
\[
\alpha \equiv \frac{A_2}{M_2}
\]
is the fraction of tokens held by the auctioneer at the beginning of period~2. The equilibrium effort is therefore:
\[
e^*=\alpha,
\]
and the expected period-2 token price is:
\[
p_2^e= \frac{k+\alpha}{M_2},
\]

Using the fact that
\[
M_2=A_2+S_1(1+\tau),
\]
We can derive the value of tokens held by investors at the start of period~2 as:
\begin{equation}\label{eq: period-2 value speculation}
S_1(1+\tau) p^e_2=(M_2-A_2) \cdot \frac{k+\alpha}{M_2} = (1-\alpha)(k+\alpha).
\end{equation}
Note, for future reference, that the effect of $\alpha$ on the value of investors' token holdings is ambiguous: a higher $\alpha$ implies higher effort and thus higher expected price, but investors hold a smaller fraction of these revenues. Specifically, expression \eqref{eq: period-2 value speculation} is strictly increasing in $\alpha$ for $\alpha<\frac{1-k}{2}$ and strictly decreasing otherwise.

While the fraction $\alpha$ is treated as given at the beginning of period~2, it is determined in period~1 by the speculative demand $S_1$. This speculative demand, in turn, depends on period-1 realized valuations and on the chosen monetary policy, as the following lemma shows.
\begin{lemma}\label{lem:effort}
For a given level of first-period revenues $\sum_i b_{i,1}$, the equilibrium value of $\alpha$ is given by:
\[
\alpha = \min\left\{ 1, \frac{\sqrt{k^2+4 \sum_i b_{i,1}(1+\sigma)}-k}{2}\right\}.
\]
\end{lemma}
Hence, for any given $\sum_i b_{i,1}$, $\alpha$ equals zero when $\sigma = -1$, is strictly increasing in $\sigma$, and reaches the upper bound of 1 whenever $\sum_i b_{i,1} > \frac{k+1}{1+\sigma}$. 
Note also that if the auctioneer sets $\sigma$ such that $\underline v >\frac{k+1}{1+\sigma}$, the resulting equilibrium coincides exactly with that of the dollar-based auction with $c=0$.


Given this, we can now write the auctioneer's payoff in period~1. Since the expected value of tokens held by investors at the start of period~2 (see equation~\eqref{eq: period-2 value speculation}) must equal the additional revenues earned by the auctioneer in period~1, the auctioneer's period-1 expected utility becomes:
\[
E[U]=  k + E[(1-\alpha)(k+\alpha)] + \beta \left(  E[\alpha (k+e(\alpha))] - \frac{E[e(\alpha)^2]}{2} \right),
\]
where expectations are taken over the realization of period-1 valuations. Differentiating this expression with respect to $\sigma$ yields:\footnote{By the envelope theorem, the effect of $\sigma$ on the optimal choice of effort can be ignored.}
\[
\frac{\partial E[U]}{\partial \sigma} = E\left[ \left(1-k-2\alpha +\beta(k+\alpha)\right)\frac{\partial \alpha}{\partial \sigma} \right].
\]

Recall from Lemma~\ref{lem:effort} that if $\sigma$ satisfies $\underline v \leq \frac{k+1}{1+\sigma}$, then $\alpha=1$ with probability one. Consequently, $\frac{\partial \alpha}{\partial \sigma}$ is zero for all realizations of $\sum_i b_{i,1}$ whenever $\underline v \leq \frac{k+1}{1+\sigma}$, but is strictly positive in expectation otherwise. These observations imply that, if $\beta <1$, the above derivative is negative whenever $\underline v$ is slightly larger than $\frac{k+1}{1+\sigma}$. Thus, the auctioneer optimally sets $\sigma$ such that $\alpha$ is strictly less than~1 in expectation. Intuitively, the auctioneer balances the trade-off between inducing higher effort and front-loading revenues by choosing a monetary policy that generates strictly positive speculative token demand on average. Since the token-based auction with $\underline v = \frac{k+1}{1+\sigma}$ exactly replicates the dollar-based auction with $c=0$, this result also implies that the optimal token-based auction is strictly preferred to the dollar-based auction with $c=0$.

The following proposition summarizes these observations.

\begin{proposition}\label{prop: with frictions}
If revenue misappropriation is sufficiently easy (i.e., $c$ is low), the auctioneer prefers the token-based auction. If revenue misappropriation is sufficiently costly (i.e., $c$ is high), the auctioneer prefers the dollar-based auction combined with a traditional financial instrument.
\end{proposition}

The key takeaway is that the possibility of misappropriating revenues influences the dollar-based auction but has no direct impact on the token-based auction. Nevertheless, if the contracting environment is sufficiently rich, the dollar-based auction, combined with an appropriately designed contingent security, outperforms the token-based auction. Conversely, if the contracting environment is sufficiently limited, the token-based auction is preferred.

\paragraph{Discussion}

The parameter $\alpha$ in the auction with tokens determines the fraction of period-2 revenue kept by the auctioneer, with the remainder accruing to investors. It plays the same role as equity in an auction with dollars. The only difference is that, with tokens, $\alpha$ is not chosen directly, because it is determined by the monetary policy selected in period 1 and by the realized period-1 valuations. By contrast, in an auction with dollars in which the auctioneer uses equity, the auctioneer can set $\alpha$ directly. Apart from this distinction, in line with the benchmark model, the auction with tokens is equivalent to an auction with dollars in which the auctioneer issues equity.  

The key observation is that, with non-contractible effort and the possibility of revenue misappropriation, equity is \textit{not} the optimal financial instrument: the first-best is achieved by a financial contract that sells future revenues to investors while imposing a penalty on the auctioneer whenever those revenues fall below a preset floor.  Such a contract is feasible only when the private cost of misappropriation, denoted by $c$, is high enough to make the penalty credible.

To conclude, note that Proposition \ref{prop: with frictions} extends to any horizon $T>2$.  For the dollar-based auction, the logic is unchanged:  
\begin{itemize}
\item If misappropriating revenues is costless ($c=0$) no outside investment is possible,
\item If $c$ is high enough, the first-best is achievable by a contract that, in each period, transfers all realized revenues to investors and penalizes the auctioneer whenever those revenues fall below a stated floor.
\end{itemize}
With respect to the token-based auction, note that the final period of the repeated auction is identical to the final period of the two-period auction discussed earlier. This last period is therefore inefficient, either because of the time-profile of revenues or because effort is below the first-best level (or both).  Hence, the same trade-off reappears: when $c$ is low the auctioneer prefers tokens because they allow to front-load part of the revenue stream, while for sufficiently large $c$ the dollar auction with its contingent security attains the first best and dominates the auction with tokens.

\section{Conclusions}

Issuing tokens and accepting them as payment allows the auctioneer to receive revenues earlier and reduces revenue risk compared to a standard repeated auction where bidders pay in dollars. Specifically, if the auctioneer commits to destroying all tokens received as payment, she realizes all revenues in the initial period and (almost) entirely eliminates risk. An identical outcome can be replicated in a traditional dollar-based auction if the auctioneer simultaneously issues equity.

The reason for this equivalence is the auctioneer's full commitment, both to deliver the object to token holders and to transfer cash revenues to equity holders. Relaxing this assumption, I introduce non-contractible effort and the possibility of revenue misappropriation. Under these conditions, the equilibrium outcomes of token-based auctions may differ from those of dollar-based auctions with financial securities. On the one hand, misappropriation is a friction that matters for dollar-based auctions but not for token-based ones. On the other hand, financial securities permit rich state-contingent payments that tokens cannot replicate. I show that when contracting frictions are severe, token-based auctions outperform dollar-based auctions; conversely, when frictions are mild or absent, dollar-based auctions combined with financial securities are preferable.

%
%

I have treated token issuance and the issuance of contingent securities as alternative mechanisms for raising funds. However, they don't have to be mutually exclusive. For instance, an auctioneer may issue tokens, sell some on the open market (to be used for bidding), and simultaneously enter into an investment contract promising the delivery of additional tokens in the future. Such a contract could specify penalties if, for example, the value of the tokens (which depends on the value of the object) falls below a given threshold. Any such punishment, however, must be incentive compatible if the auctioneer has the option to ``run away''. Exploring this possibility is left for future work.

Throughout the analysis, I have assumed that the market for tokens is frictionless. In this setting, holding tokens is only inconvenient in that consumption may be delayed by one period, but this delay entails no cost. A more realistic approach would recognize frictions in the token market, which would further diminish the benefits of using tokens. In an even richer framework, these frictions could depend on transaction volume, which in turn would be linked to the value of the object being traded. Extending the model in this direction is also left for future work.


\appendix

\section{Mathematical derivations}\label{appendix-mathematical derivations}

\begin{proof}[Proof of Lemma \ref{lem: initial}]
Consider period $t$. Take the total payments to the auctioneer $\sum_i b_{i,t}$ and the expected future price $p^e_{t+1}$ as given. Note that a bidder can spend 1 dollar to purchase $\frac{1}{p_t}$ tokens in period $t$, these tokens then multiply by $1+\tau_t$ and can be sold in period $t+1$,  for a present-discounted return of $\frac{\beta (1+\tau_t)p^e_{t+1}}{p_t}$. Alternatively, he can invest the same amount of money in the risk-free asset for a present-discounted return of $\beta R =1$ in the following period. It follows that there can be an equilibrium in the market for tokens if and only if $p_t \geq \beta (1+\tau_t)p^e_{t+1}$. 

If $p_t > \beta (1+\tau_t)p^e_{t+1}$, no tokens are purchased and then brought to the next period. The demand for tokens is given by the tokens used for bidding $\frac{\sum_i b_{i,t}}{p_t}$. The supply of tokens is $M_t$, and hence the equilibrium  price  is 
\[
p_t = \frac{\sum_i b_{i,t}}{M_t}
\]
This is indeed an equilibrium if $\frac{\sum_i b_{i,t}}{M_t}>\beta (1+\tau_t) p^e_{t+1}$, that is, if the realized bids are sufficiently high relative to the future expected price. 

If instead $\frac{\sum_i b_{i,t}}{M_t} \leq \beta (1+\tau_t) p^e_{t+1}$, then tokens may be purchased but not used for bidding.  I call this demand the speculative demand for tokens $S_t$. The total demand for tokens is now $S_t+\frac{\sum_i b_{i,t}}{p_t}$, and the equilibrium price is 
\[
p_t = \frac{\sum_i b_{i,t}}{M_t-S_t}=\beta (1+\tau_t) p^e_{t+1}
\]
which pins down  the speculative demand for tokens:
\[
S_t = \max\left\lbrace M_t-\frac{\sum_i^n b_{i,t}}{\beta (1+\tau_t)p^e_{t+1}},0  \right\rbrace
\]
\end{proof}

\begin{proof}[Proof of Proposition \ref{prop:revenues}]
To start, note the following two facts:
\begin{enumerate}
\item in each period the auctioneer will liquidate all his tokens. If the price in a given period is such that $p_t>\beta (1+\tau_t) p^e_{t+1}$, then the auctioneer is better off selling tokens in period $t$ at a higher price than in period $t+1$ at a lower price. If instead $p_t=\beta (1+\tau_t) p^e_{t+1}$ and the auctioneer is risk averse, then again he prefers to earn revenues in period $t$ than to wait one period and earn the same expected revenues (but this time being exposed to risk). If $p_t=\beta (1+\tau_t) p^e_{t+1}$ and the auctioneer is risk neutral then his expected revenues are the same whether he holds tokens between periods (i.e., he acts as a speculator) or not. Without loss of generality, we can assume that he sells all his tokens in period $t$ also in this case;
\item an implication of the above  fact is that, for given $p^e_{t+1}$, $A_{t+1}$ is independent of $A_t$. That is because $A_{t+1}$ depends on the monetary policy parameters $\sigma_t, \tau_t$ and on the tokens received for payment in period $t$, which depend exclusively on $p^e_{t+1}$ and the realization of valuation in period $t$.
\end{enumerate}

I can therefore write the auctioneer's expected continuation revenues in period $t$  recursively as:
\[
E[\Pi_t(A_t)]=p^e_t A_t + \beta E[p_{t+1} A_{t+1} ]+  \beta^2 E[\Pi_{t+2}(A_{t+2})]
\]
where the expectations are taken at the beginning of period $t$, before the valuations are realized. 

Consider now a given  $p^e_{t+1}$. The realization of valuations in period $t$ may be such that $S_t=0$ and hence $p_t=\frac{\sum_i b_{i,t}}{M_t}$. In this case, $A_{t+1}=M_{t+1}$, so that, conditional on $S_t=0$:
\begin{footnotesize}
\begin{align*}
E[\Pi_t(A_t)|S_t=0]= &E[p_t|S_t=0]  A_t + \beta E[p_{t+1} A_{t+1}|S_t=0] +  \beta^2 E[\Pi_{t+2}(A_{t+2})|S_t=0] \\
=& E\left[ \frac{\sum_i b_{i,t}}{M_t} |S_t=0 \right] (M_t-M_t +A_t) +  \beta p^e_{t+1} E[  M_{t+1}|S_t=0] +  \beta^2 E[\Pi_{t+2}(A_{t+2})|S_t=0] 
\\ =&E[\sum_i b_{i,t}|S_t=0] - E[p_t|S_t=0] (M_t-A_t)  +  \beta p^e_{t+1} E[  M_{t+1}|S_t=0] +  \beta^2 E[\Pi_{t+2}(A_{t+2})|S_t=0] 
\end{align*}
\end{footnotesize}
If instead $S_t>0$, then $p_t=\beta  (1+\tau)  p^e_{t+1}$ and  $A_{t+1}=M_{t+1}-S_t (1+\tau_t)\geq 0$. In this case,  conditional on $S_t>0$:
\begin{footnotesize}
\begin{align*}
E&[\Pi_t(A_t)|S_t>0]=  E[p_t|S_t>0]  A_t + \beta E[p_{t+1} A_{t+1}|S_t>0] + \beta^2 E[\Pi_{t+2}(A_{t+2})|S_t>0]  \\ &= E[p_t|S_t>0]  A_t  +  \beta p^e_{t+1} (E[M_{t+1}|S_t>0] -E[S_t|S_t>0] (1+\tau_t)) +  \beta^2 E[\Pi_{t+2}(A_{t+2})|S_t>0]  \\
 &=  E[p_t|S_t>0]  A_t  +  \beta p^e_{t+1}  E[M_{t+1}|S_t>0]- \beta p^e_{t+1} (1+\tau_t) \left( M_t-\frac{E[\sum_i b_{i,t}|S_t>0]}{\beta p^e_{t+1} (1+\tau_t)} \right) +  \beta^2 E[\Pi_{t+2}(A_{t+2})|S_t>0] \\
& =  E\left[ \sum_i b_{i,t}|   S_t>0 \right]  -  E[p_t|S_t>0] (M_t-A_t)   + \beta p^e_{t+1} E[M_{t+1}|S_t>0] +  \beta^2 E[\Pi_{t+2}(A_{t+2})|S_t>0]
\end{align*}
\end{footnotesize}
where I used the definition of $S_t$ as well as the fact that $E[p_t|S_t>0] = \beta p^e_{t+1} (1+\tau_t)$.

The above derivations then imply that, for a given sequence of expected equilibrium prices, the unconditional expected revenues are (again, all expectations are taken at the beginning of period $t$):
\begin{align*}
E[\Pi_t(A_t)]&=k - p^e_t (M_t-A_t)  + \beta p^e_{t+1} E[M_{t+1}] +  \beta^2 E[\Pi_{t+2}(A_{t+2})]\\& = k - p^e_t (M_t-A_t) +\beta  E[\Pi_{t+1}(M_{t+1})]
\end{align*}
where the last equality exploits the fact that, for given equilibrium expected prices, $A_{t+2}$ is independent of $A_{t+1}$ (see point 2 above). The statement then follows by iterating the above equation.

\end{proof}

\

\begin{proof}[Proof of Lemma \ref{lem:effort}]
By using the definition of $\alpha$, I can write
\begin{equation}\label{eq:effort-in-proof}
\alpha = \frac{(1-S_1)(1+\sigma)(1+\tau) }{(1-S_1)(1+\sigma)(1+\tau)+S_1(1+\tau)}
\end{equation}
Next, I use the above expression to write
\begin{equation}\label{eq:p2-in-proof}
p^e_2=\frac{k+\alpha}{(1-S_1)(1+\sigma)(1+\tau)+S_1(1+\tau)}=\frac{(k+\alpha)\alpha}{(1-S_1)(1+\sigma)(1+\tau)}
\end{equation}
I use the definition of $S_1$ and the above expression to obtain:
\[
S_1=\max\left\lbrace 0, 1-\frac{\sum_i b_{i,1}}{p^e_2 (1+\tau)} \right\rbrace = \max\left\lbrace 0, 1-\frac{\sum_i b_{i,1}}{(k+\alpha)\alpha}  (1-S_1)(1+\sigma) \right\rbrace
\]
\[
1-S_1= \min \left\lbrace 1, \frac{\sum_i b_{i,1}}{(k+\alpha)\alpha} (1-S_1)(1+\sigma) \right\rbrace
\]
Suppose $\alpha<1$ so that $0<S_1\leq 1$. Then $\alpha$ must be such that 
\[
\alpha = \sqrt{\left( \frac{k}{2}\right)^2+(\sum_i b_{i,1})(1+\sigma)}-\frac{k}{2}  
\]
this can be the solution as long as $\alpha<1$, or $k+1>\sum_i b_{i,1} (1+\sigma)$.
If instead $\alpha=1$, then $S_1=0$ and it must be that 
\[
k+1<\sum_i b_{i,1} (1+\sigma)
\]
Putting everything together, we have that 
\[
\alpha = \min \left\lbrace 1, \sqrt{\left( \frac{k}{2}\right)^2+\sum_i b_{i,1} (1+\sigma)}-\frac{k}{2}  \right\rbrace.
\]
\end{proof}


\bibliography{bibliography}{}

\begin{thebibliography}{}

\bibitem[\protect\citeauthoryear{Bakos and Halaburda}{Bakos and Halaburda}{2019}]{bakos2019funding}
Bakos, Y. and H.~Halaburda (2019).
\newblock Funding new ventures with digital tokens: due diligence and token tradability.
\newblock {\em NYU Stern School of Business\/}.

\bibitem[\protect\citeauthoryear{Bakos and Halaburda}{Bakos and Halaburda}{2022}]{bakos2018role}
Bakos, Y. and H.~Halaburda (2022).
\newblock Overcoming the coordination problem in new marketplaces via cryptographic tokens.
\newblock {\em Information Systems Research\/}~{\em 33\/}(4), 1368--1385.

\bibitem[\protect\citeauthoryear{Brzustowski, Georgiadis, and Szentes}{Brzustowski et~al.}{2021}]{brzustowski2021smart}
Brzustowski, T., A.~Georgiadis, and B.~Szentes (2021).
\newblock Smart contracts and the coase conjecture.
\newblock Technical report, Working Paper.

\bibitem[\protect\citeauthoryear{Budish}{Budish}{2011}]{budish2011combinatorial}
Budish, E. (2011).
\newblock The combinatorial assignment problem: Approximate competitive equilibrium from equal incomes.
\newblock {\em Journal of Political Economy\/}~{\em 119\/}(6), 1061--1103.

\bibitem[\protect\citeauthoryear{Budish, Cachon, Kessler, and Othman}{Budish et~al.}{2017}]{budish2017course}
Budish, E., G.~P. Cachon, J.~B. Kessler, and A.~Othman (2017).
\newblock Course match: A large-scale implementation of approximate competitive equilibrium from equal incomes for combinatorial allocation.
\newblock {\em Operations Research\/}~{\em 65\/}(2), 314--336.

\bibitem[\protect\citeauthoryear{Bulow and Klemperer}{Bulow and Klemperer}{1996}]{bulow1996auctions}
Bulow, J. and P.~Klemperer (1996).
\newblock Auctions versus negotiations.
\newblock {\em The American Economic Review\/}~{\em 86\/}(1), 180--194.

\bibitem[\protect\citeauthoryear{Bulow and Roberts}{Bulow and Roberts}{1989}]{bulow1989simple}
Bulow, J. and J.~Roberts (1989).
\newblock The simple economics of optimal auctions.
\newblock {\em Journal of political economy\/}~{\em 97\/}(5), 1060--1090.

\bibitem[\protect\citeauthoryear{Canidio}{Canidio}{2018}]{canidio2018financial}
Canidio, A. (2018).
\newblock Financial incentives for the development of blockchain-based platforms.

\bibitem[\protect\citeauthoryear{Canidio}{Canidio}{2020}]{canidio2020cryptotokens}
Canidio, A. (2020).
\newblock Cryptotokens and cryptocurrencies: the extensive margin.
\newblock {\em working paper\/}.

\bibitem[\protect\citeauthoryear{Canidio}{Canidio}{2023}]{canidio2022Auction-pp}
Canidio, A. (2023).
\newblock Financial bubbles in infinitely repeated auctions with tokens.
\newblock In {\em AEA Papers and Proceedings}, Volume 113, pp.\  263--67.

\bibitem[\protect\citeauthoryear{Canidio and Danos}{Canidio and Danos}{2022}]{canidio-danos}
Canidio, A. and V.~Danos (2022).
\newblock Commitment against front running attacks.
\newblock {\em working paper\/}.

\bibitem[\protect\citeauthoryear{Capponi and Jia}{Capponi and Jia}{2021}]{capponi2021adoption}
Capponi, A. and R.~Jia (2021).
\newblock The adoption of blockchain-based decentralized exchanges.
\newblock {\em arXiv preprint arXiv:2103.08842\/}.

\bibitem[\protect\citeauthoryear{Catalini and Gans}{Catalini and Gans}{2018}]{RePEc:nbr:nberwo:24418}
Catalini, C. and J.~S. Gans (2018, March).
\newblock {Initial Coin Offerings and the Value of Crypto Tokens}.
\newblock NBER Working Papers 24418, National Bureau of Economic Research, Inc.

\bibitem[\protect\citeauthoryear{Chod and Lyandres}{Chod and Lyandres}{2021}]{doi:10.1287/mnsc.2020.3754}
Chod, J. and E.~Lyandres (2021).
\newblock A theory of icos: Diversification, agency, and information asymmetry.
\newblock {\em Management Science\/}~{\em 67\/}(10), 5969--5989.

\bibitem[\protect\citeauthoryear{Cong, Li, and Wang}{Cong et~al.}{2021}]{cong2018tokenomics}
Cong, L.~W., Y.~Li, and N.~Wang (2021).
\newblock Tokenomics: Dynamic adoption and valuation.
\newblock {\em The Review of Financial Studies\/}~{\em 34\/}(3), 1105--1155.

\bibitem[\protect\citeauthoryear{Cong, Li, and Wang}{Cong et~al.}{2022}]{cong2020token}
Cong, L.~W., Y.~Li, and N.~Wang (2022).
\newblock Token-based platform finance.
\newblock {\em Journal of Financial Economics\/}~{\em 144\/}(3), 972--991.

\bibitem[\protect\citeauthoryear{Gans}{Gans}{2023}]{Gans2023}
Gans, J. (2023).
\newblock {\em Cryptography Versus Incentives}, pp.\  85--101.
\newblock Cham: Springer International Publishing.

\bibitem[\protect\citeauthoryear{Gans}{Gans}{2019}]{gans2019fine}
Gans, J.~S. (2019).
\newblock The fine print in smart contracts.
\newblock Technical report, National Bureau of Economic Research.

\bibitem[\protect\citeauthoryear{Gans and Holden}{Gans and Holden}{2023}]{gans2023solomonic}
Gans, J.~S. and R.~Holden (2023).
\newblock A solomonic solution to blockchain front-running.
\newblock In {\em AEA Papers and Proceedings}, Volume 113, pp.\  248--52.

\bibitem[\protect\citeauthoryear{Gans and Holden}{Gans and Holden}{2022}]{gans2022mechanism}
Gans, J.~S. and R.~T. Holden (2022).
\newblock Mechanism design approaches to blockchain consensus.
\newblock Technical report, National Bureau of Economic Research.

\bibitem[\protect\citeauthoryear{Garratt and {van Oordt}}{Garratt and {van Oordt}}{2021}]{GARRATT2021104171}
Garratt, R.~J. and M.~R. {van Oordt} (2021).
\newblock Entrepreneurial incentives and the role of initial coin offerings.
\newblock {\em Journal of Economic Dynamics and Control\/}, 104171.

\bibitem[\protect\citeauthoryear{Goldstein, Gupta, and Sverchkov}{Goldstein et~al.}{2019}]{goldstein2019initial}
Goldstein, I., D.~Gupta, and R.~Sverchkov (2019).
\newblock Utility tokens as a commitment to competition.
\newblock {\em Available at SSRN 3484627\/}.

\bibitem[\protect\citeauthoryear{Gryglewicz, Mayer, and Morellec}{Gryglewicz et~al.}{2021}]{gryglewicz2021optimal}
Gryglewicz, S., S.~Mayer, and E.~Morellec (2021).
\newblock Optimal financing with tokens.
\newblock {\em Journal of Financial Economics\/}~{\em 142\/}(3), 1038--1067.

\bibitem[\protect\citeauthoryear{He, Miralles, Pycia, and Yan}{He et~al.}{2018}]{he2018pseudo}
He, Y., A.~Miralles, M.~Pycia, and J.~Yan (2018).
\newblock A pseudo-market approach to allocation with priorities.
\newblock {\em American Economic Journal: Microeconomics\/}~{\em 10\/}(3), 272--314.

\bibitem[\protect\citeauthoryear{Holden and Malani}{Holden and Malani}{2021}]{holden_malani_2021}
Holden, R. and A.~Malani (2021).
\newblock {\em Can Blockchain Solve the Hold-up Problem in Contracts?}
\newblock Elements in Law, Economics and Politics. Cambridge University Press.

\bibitem[\protect\citeauthoryear{Hylland and Zeckhauser}{Hylland and Zeckhauser}{1979}]{hylland1979efficient}
Hylland, A. and R.~Zeckhauser (1979).
\newblock The efficient allocation of individuals to positions.
\newblock {\em Journal of Political economy\/}~{\em 87\/}(2), 293--314.

\bibitem[\protect\citeauthoryear{Klemperer}{Klemperer}{1999}]{klemperer1999auction}
Klemperer, P. (1999).
\newblock Auction theory: A guide to the literature.
\newblock {\em Journal of economic surveys\/}~{\em 13\/}(3), 227--286.

\bibitem[\protect\citeauthoryear{Kocherlakota}{Kocherlakota}{1998}]{kocherlakota1998money}
Kocherlakota, N.~R. (1998).
\newblock Money is memory.
\newblock {\em journal of economic theory\/}~{\em 81\/}(2), 232--251.

\bibitem[\protect\citeauthoryear{Lee, Martin, and Townsend}{Lee et~al.}{2021}]{lee2021optimal}
Lee, M., A.~Martin, and R.~M. Townsend (2021).
\newblock Optimal design of tokenized markets.
\newblock {\em Available at SSRN 3820973\/}.

\bibitem[\protect\citeauthoryear{Li and Mann}{Li and Mann}{2018}]{li2018initial}
Li, J. and W.~Mann (2018).
\newblock Digital tokens and platform building.
\newblock {\em Working paper\/}.

\bibitem[\protect\citeauthoryear{Malinova and Park}{Malinova and Park}{2018}]{malinova2018tokenomics}
Malinova, K. and A.~Park (2018).
\newblock Tokenomics: when tokens beat equity.
\newblock {\em Available at SSRN 3286825\/}.

\bibitem[\protect\citeauthoryear{Milgrom and Segal}{Milgrom and Segal}{2002}]{milgrom2002envelope}
Milgrom, P. and I.~Segal (2002).
\newblock Envelope theorems for arbitrary choice sets.
\newblock {\em Econometrica\/}~{\em 70\/}(2), 583--601.

\bibitem[\protect\citeauthoryear{Myerson}{Myerson}{1981}]{myerson1981optimal}
Myerson, R.~B. (1981).
\newblock Optimal auction design.
\newblock {\em Mathematics of operations research\/}~{\em 6\/}(1), 58--73.

\bibitem[\protect\citeauthoryear{Prat, Danos, and Marcassa}{Prat et~al.}{2019}]{prat2019fundamental}
Prat, J., V.~Danos, and S.~Marcassa (2019).
\newblock Fundamental pricing of utility tokens.
\newblock {\em working paper\/}.

\bibitem[\protect\citeauthoryear{Prendergast}{Prendergast}{2017}]{prendergast2017food}
Prendergast, C. (2017).
\newblock How food banks use markets to feed the poor.
\newblock {\em Journal of Economic Perspectives\/}~{\em 31\/}(4), 145--62.

\bibitem[\protect\citeauthoryear{Prendergast}{Prendergast}{2022}]{prendergast2021food}
Prendergast, C. (2022).
\newblock The allocation of food to food banks.
\newblock {\em Journal of Political Economy\/}~{\em 130\/}(8), 1993--2017.

\bibitem[\protect\citeauthoryear{Riley and Samuelson}{Riley and Samuelson}{1981}]{riley1981optimal}
Riley, J.~G. and W.~F. Samuelson (1981).
\newblock Optimal auctions.
\newblock {\em The American Economic Review\/}~{\em 71\/}(3), 381--392.

\bibitem[\protect\citeauthoryear{Roughgarden}{Roughgarden}{2020}]{roughgarden2020transaction}
Roughgarden, T. (2020).
\newblock Transaction fee mechanism design for the ethereum blockchain: An economic analysis of eip-1559.
\newblock {\em arXiv preprint arXiv:2012.00854\/}.

\bibitem[\protect\citeauthoryear{Sockin and Xiong}{Sockin and Xiong}{2023}]{sockin2018model}
Sockin, M. and W.~Xiong (2023).
\newblock A model of cryptocurrencies.
\newblock {\em Management Science\/}.

\bibitem[\protect\citeauthoryear{Townsend}{Townsend}{2020}]{townsend2020distributed}
Townsend, R.~M. (2020).
\newblock {\em Distributed Ledgers: Design and Regulation of Financial Infrastructure and Payment Systems}.
\newblock MIT Press.

\bibitem[\protect\citeauthoryear{Townsend and Zhang}{Townsend and Zhang}{2023}]{townsend-zhang}
Townsend, R.~M. and N.~X. Zhang (2023).
\newblock Technologies that replace a" central planner".
\newblock In {\em AEA Papers and Proceedings}, Volume 113, pp.\  257--62. American Economic Association.

\bibitem[\protect\citeauthoryear{Vickrey}{Vickrey}{1961}]{vickrey1961counterspeculation}
Vickrey, W. (1961).
\newblock Counterspeculation, auctions, and competitive sealed tenders.
\newblock {\em The Journal of finance\/}~{\em 16\/}(1), 8--37.

\end{thebibliography}
\bibliographystyle{chicago}

\end{document}